\documentclass[12pt]{article}
\usepackage{amsmath,amsfonts, amsthm,amssymb,bbm}

\newtheorem*{ex*}{Example(s)}
\newtheorem*{df*}{Definition}
\newtheorem*{thm*}{Theorem}
\newtheorem*{prop*}{Proposition}

\newtheorem{thm}{Theorem}

\newtheorem{lemma}{Lemma}[section]

\newtheorem{prop}[lemma]{Proposition}
\newtheorem{rmk}[lemma]{Remark}

\begin{document}

\title{Positive temperature versions of two theorems on first-passage percolation}

\author{Sasha Sodin\footnote{Department of Mathematics, Princeton University, Princeton,
NJ 08544, USA. E-mail: asodin@princeton.edu.
Supported in part by NSF grant PHY-1305472.}}

\maketitle 

\begin{abstract}
The estimates on the fluctuations of first-passsage percolation due to Ta\-la\-grand (a tail bound)
and Ben\-ja\-mini--Kalai--Schramm (a sublinear
variance bound) are transcribed into the posi\-tive-temperature setting
of random Schr\"odinger operators.
\end{abstract}

\section{Introduction}

Let $H = - \frac{1}{2d} \Delta + V$ be a random Schr\"odinger operator on $\mathbb{Z}^d$
with non-negative potential $V \geq 0$:
\[ (H\psi)(x) = (1 + V(x))\psi(x) - \frac{1}{2d} \sum_{y \sim x} \psi(y)~, \quad \psi \in \ell^2(\mathbb{Z}^d)~.\]
Assume that the entries of $V$ are independent, identically distributed, and satisfy 
\begin{equation}\label{eq:not0}
\mathbb{P}\{V(x) > 0\} > 0~. 
\end{equation}
The inverse $G = H^{-1}$ of $H$ defines a random metric 
\begin{equation}\label{eq:def.rho}
\rho(x, y) = \log \frac{\sqrt{G(x,x)G(y,y)}}{G(x,y)}
\end{equation}
on $\mathbb{Z}^d$ (see Lemma~\ref{l:bound.infty} below for the verification of the triangle inequality).
We are interested in the  behaviour of $\rho(x, y)$ for large $\| x - y \|$ (here and forth
$\| \cdot \|$ stands for the $\ell_1$ norm); to simplify the notation, set $\rho(x) = \rho(0, x)$.

Zerner proved \cite[Theorem~A]{Z}, using Kingman's subadditive ergodic theorem \cite{Ki}, that if $V$ satisfies (\ref{eq:not0})
and
\begin{equation}\label{eq:z.moment}
\mathbb{E} \log^d (1 + V(x)) < \infty~.
\end{equation}
then
\begin{equation}\label{eq:king}
\rho(x) = \|x \|_V (1+ o(1))~, \quad \|x \| \to \infty~, 
\end{equation}
where $\| \cdot \|_V$ is a deterministic norm on $\mathbb{R}^d$ determined by the distribution of $V$. 
As to the fluctuations of $\rho(x)$, Zerner showed \cite[Theorem~C]{Z} that (\ref{eq:not0}),
(\ref{eq:z.moment}), and
\[ \text{if $d = 2$, then $\mathbb{P} \left\{ V(x) = 0 \right\} = 0$} \]
imply the bound
\begin{equation}\label{eq:z}
\mathrm{Var} \, \rho(x) \leq C_V \|x\|~.
\end{equation}
In dimension $d=1$, the bound (\ref{eq:z}) is sharp; moreover, $\rho$ obeys a central limit theorem
\[ \frac{\rho(x) - \mathbb{E} \rho(x)}{\sigma_V |x|^{1/2}} \underset{|x| \to \infty}{\overset{D}{\longrightarrow}} N(0, 1)~,\]
which follows from the results of Furstenberg and Kesten \cite{FK}. In higher dimension, the fluctuations of $\rho$
are expected to be smaller: the fluctuation exponent
\[ \chi = \limsup_{\|x\| \to \infty} \frac{\frac{1}{2} \log \mathrm{Var} \, \rho(x)}{\log \|x\|} \]
is expected to be equal to $1/3$ in dimension $d=2$, and to be even smaller for $d \geq 3$; see
Krug and Spohn \cite{KSp}. 

These conjectures are closely related to the corresponding conjectures for first-passage percolation. In fact,
$\rho$ is a positive-temperature counterpart of the (site) first-passage percolation metric corresponding to
$\omega = \log(1+V)$; we refer to Zerner \cite[Section~3]{Z} for a more elaborate discussion of this connection. 

The rigorous understanding of random metric fluctuations in dimension $d \geq 2$ is for now confined to a handful of (two-dimensional) integrable models, where $\chi = 1/3$
(see Corwin \cite{C} for a review), and to several weak-disorder models in dimension $d \geq 4$, for which $\chi = 0$ (see Imbrie and Spencer \cite{IS}, and Bolthausen \cite{Bolt}). 

It is a major problem to find the value of the exponent
$\chi$ beyond these two classes of models. 
We refer to the works of  Chatterjee \cite{Ch} and Auffinger--Damron \cite{AD1,AD} for some recent results (in arbitrary dimension) establishing a connection between the fluctuation exponent $\chi$ and
the wandering exponent $\xi$
describing transversal fluctuations of the geodesics.

\vspace{2mm}\noindent
Here we carry out a much more modest task: verifying that the upper bounds on the fluctuations in (bond)  first-passage
percolation due to Talagrand \cite{T} and Benjamini--Kalai--Schramm \cite{BKS} are also valid for the random metric
(\ref{eq:def.rho}). Zerner's bound (\ref{eq:z}) is a positive-temperature counterpart of Kesten's estimate \cite{K}. 
Kesten showed that the (bond) first-passage percolation $\rho_\text{FPP}$ satisfies
\begin{equation} \mathrm{Var} \, \rho_\text{FPP}(x) \leq C \| x\|~; 
\end{equation}
furthermore, if the underlying random variables have exponential tails, then so does 
$(\rho_\text{FPP}(x) - \mathbb{E}\rho_\text{FPP}(x)) / \sqrt{\|x\|}$.
Talagrand improved the tail bound to 
\[ \mathbb{P} \left\{ | \rho_\text{FPP}(x) - \mathbb{E} \rho_\text{FPP}(x)| \geq t \right\} \leq C \exp \left\{ - \frac{t^2}{C\|x\|} \right\}~,
    \quad 0 \leq t \leq \|x\|~.\]
Benjamini, Kalai, and Schramm \cite{BKS} proved, in dimension $d \geq 2$, the sublinear bound 
\begin{equation}\label{eq:subl} \mathrm{Var} \, \rho_\text{FPP}(x) \leq C \|x\| / \log (\|x\|+2)~, 
\end{equation}
for the special case of Bernoulli-distributed potential. Bena\"im and  Rossignol \cite{BR} extended this bound to 
a wider class of distributions (``nearly gamma'' in the terminology of \cite{BR}), and complemented it 
with an exponential tail estimate. Damron, Hanson, and Sosoe \cite{DHS}
proved (\ref{eq:subl}) for arbitrary potential with $2+\log$ moments.
Extensions of the Benjamini--Kalai--Schramm bound to other models have been found by van der Berg and Kiss \cite{vdBK}, Matic and Nolen \cite{MN},
and Alexander and Zygouras \cite{AZ}.

\vspace{2mm}\noindent
Theorem~\ref{thm:fluct} below is a positive temperature analogue of Talagrand's bound
(in order to use a more elementary concentration inequality from \cite{T0,T} instead of a more
involved one from \cite{T}, we establish a slightly stronger conclusion
under a slightly stronger assumption). Theorem~\ref{thm:bks} is a positive temperature analogue of the Benjamini--Kalai--Schramm bound.

 The strategy of the proof 
is very close to the original arguments; the modification mainly enters in a couple of deterministic estimates. Compared to the closely related work of Piza \cite{P} on directed polymenrs, we economise on the use of the random walk representation (\ref{eq:rw}), with the hope that the savings will eventually suffice to address an extension discussed in Section~\ref{s:rem}.

Set $\mu(x) = \mathbb{E} \rho(x)$.

\begin{thm}\label{thm:fluct}
Suppose the entries of $V$ are independent, identically distri\-bu\-ted,  bounded from below by $\epsilon > 0$, and from above by $0 < M < \infty$. Then
\begin{equation}\label{eq:ltail}
 \mathbb{P} \left\{ \rho(x) \leq \mu(x) - t \right\}
    \leq C \exp \left\{ - \frac{t^2}{C(\epsilon,M) (\mu(x) + 1)} \right\}~, 
\end{equation}
 and
\begin{equation}\label{eq:rtail}
\mathbb{P} \left\{ \rho(x) \geq \mu(x) + t \right\}
    \leq C \exp \left\{ - \frac{t^2}{C(\epsilon,M) (\mu(x) + t + 1)} \right\}~, 
\end{equation}
for every $t \geq 0$.
\end{thm}

\begin{rmk}
The assumption $\epsilon \leq V \leq M$ yields the deterministic estimate 
\begin{equation}\label{eq:determ} C_\epsilon^{-1} \|x\| \leq \rho(x) \leq C_M \|x\|~, \end{equation}
which, in conjunction with (\ref{eq:ltail}) and (\ref{eq:rtail}),
implies the inequality
\[ \mathbb{P} \left\{ \left| \rho(x) - \mu(x) \right| \geq t \right\} \leq C \exp \left\{ - \frac{t^2}{C(\epsilon,M)\|x\|} \right\}~.\]
\end{rmk}

\begin{thm}\label{thm:bks}
Assume that the distribution of the potential is given by
\[ \mathbb{P} \left\{ V(x) = a \right\} = \mathbb{P} \left\{ V(x) = b \right\} =  1/2 \]
for some $0 < a < b$, and that $d \geq 2$. Then
\begin{equation}\label{eq:bks}
\mathrm{Var} \, \rho(x) \leq C_{a,b} \, \frac{\|x\|}{\log (\|x\|+2)}~. 
\end{equation}
\end{thm}

We conclude the introduction with a brief comment on lower bounds. In dimension $d=2$, Newman and Piza \cite{NP} proved the logarithmic lower bound 
\begin{equation}\label{eq:np}
\operatorname{Var} \rho_\text{FPP} (x) \geq \frac{1}{C} \log (\|x\|+1)~.
\end{equation}
A version for  directed polymers (the positive temperature counterpart of directed first passage percolation)
was proved by Piza \cite{P}; the argument there is equally applicable to the undirected polymers which
are the subject of the current note. We are not aware of
any non-trivial lower bounds in dimension $d \geq 3$.

\section{Proof of Theorem~\ref{thm:fluct}}

\vspace{3mm}\noindent
The proof of Theorem~\ref{thm:fluct} is based on Talagrand's concentration inequality \cite{T0,T}. 
We state this inequality as

\begin{lemma}[Talagrand]\label{l:radek}
Assume that $\left\{ V(x) \, \mid \, x \in \mathcal{X} \right\}$ are independent random variables, the distribution of every
one of which is supported in $[0, M]$. Then, for every convex (or concave) $L$-Lipschitz function $f: \mathbb{R}^\mathcal{X} \to \mathbb{R}$.
\[  \mathbb{P} \left\{ f \geq \mathbb{E} f + t \right\} \leq  C \exp \left\{ - \frac{t^2}{C M^2 L^2} \right\}~,\]
where $C>0$ is a constant.
\end{lemma}

Denote $g(x) = G(0, x)$. To apply Lemma~\ref{l:radek}, we first compute the gradient of $\log g$,
and then estimate its norm.
\begin{lemma}\label{l:deriv} For any $x, y \in \mathbb{Z}^d$,
\[ \frac{\partial}{\partial V(y)} \log g(x) = - \frac{G(0, y)G(y,x)}{G(0, x)}~. \] 
\end{lemma}
\begin{proof}
Let $P_y = \delta_y \delta_y^*$ be the projector on the $y$-th coordinate. Set
$H_h = H + h P_y$,  $G_h = H_h^{-1}$. By the resolvent identity
\[ G_h = G - h G P_y G_h~,\]
hence
\[ \frac{d}{dh} \Big|_{h = 0} G_h = - G P_y G\]
and
\[ \frac{d}{dh} \Big|_{h = 0} G_h(0, x) = - G(0, y) G(y, x)~. \]
\end{proof}

\vspace{3mm}\noindent
Our next goal is to prove 
\begin{prop}\label{p:bound.2}
Suppose $V \geq \epsilon > 0$. Then 
\begin{equation}\label{eq:l2} \sum_{y} \left[ \frac{G(0, y) G(y, x)}{G(0, x)} \right]^2 \leq
  A_\epsilon (\rho(x) + 1)~, 
\end{equation}
where $A_\epsilon$ depends only on $\epsilon$.
\end{prop}

The proof consists of two ingredients. The first one, equivalent to the triangle inequality for
$\rho$, yields an upper bound on every term in the left-hand side of (\ref{eq:l2}).
\begin{lemma}\label{l:bound.infty}
For any $x,y \in \mathbb{Z}^d$,
\[ \frac{G(0, y)G(y, x)}{G(0,x)} \leq G(y, y) \leq C_\epsilon~. \]
\end{lemma}
\begin{proof}
Let $H_y$ be the operator obtained by erasing the edges that connect $y$ to its
neighbours, and let $G_y = H_y^{-1}$. By the resolvent identity,
\[ G(0, x) = G_y(0, x) + \frac{1}{2d} \sum_{y' \sim y} G_y(0, y') G(y, x)~.\]
In particular,
\[ G(0, y) = \frac{1}{2d} \sum_{y' \sim y} G_y(0, y') G(y, y)~. \]
Therefore
\[ G(0, x) = G_y(0, x) + \frac{G(0, y) G(y, x)}{G(y, y)}~.\]
\end{proof}

The second ingredient is
\begin{lemma}\label{l:7.5}
For any $x \in \mathbb{Z}^d$,
\[ \sum_y \frac{G(0,y) G(y, x)}{G(0,x)} \leq C_\epsilon(\rho(x) + 1)~. \]
\end{lemma}
The proof of Lemma~\ref{l:7.5} requires two more lemmata. Denote
\[   g_2(x) = G^2(0, x) = \sum_y G(0, y) G(y, x)~, \quad u(x) = \frac{g_2(x)}{g(x)}~. \]

\begin{lemma}\label{l:equ} For any $x \in \mathbb{Z}^d$,
\begin{equation}\label{eq:relg}
\sum_{y \sim x} \frac{g(y)}{2d(1+V(x))g(x)} = 1 - \frac{\delta(x)}{(1+V(0))g(0)} 
\end{equation}
and
\begin{equation}\label{eq:relu}
u(x) = \sum_{y \sim x} u(y) \, \frac{g(y)}{2d(1+V(x))g(x)} + \frac{1}{1+V(x)}~. 
\end{equation}
\end{lemma}
\begin{proof}
The first formula follows from the relation $Hg = \delta$, and the second one -- from the relation $H g_2 = g$.
\end{proof}

Set $\widetilde{\rho}(x) = \log \frac{G(0,0)}{G(0, x)}$. 
\begin{lemma}\label{l:convlog}
For any $x \in \mathbb{Z}^d$,
\begin{multline*}
\widetilde{\rho}(x) \geq \sum_{y \sim x} \widetilde{\rho}(y)  \, \frac{g(y)}{2d(1+V(x))g(x)} \\
+ \log(1 + V(x)) + \log\left(1 - \frac{1}{(1+V(0))g(0)}\right) \delta(x)~. 
\end{multline*}
\end{lemma}
\begin{proof}
For $x \neq 0$, (\ref{eq:relg}) and the concavity of logarithm yield
\[ \sum_{y \sim x} \frac{g(y)}{2d(1+V(x))g(x)} \log \frac{2d(1+V(x))g(x)}{g(y)} \leq \log(2d)~. \]
Using (\ref{eq:relg}) once again, we obtain
\[ - \widetilde{\rho}(x) + \sum_{y \sim x} \widetilde{\rho}(y) \, \frac{g(y)}{2d(1+V(x))g(x)}+ \log(1 + V(x)) \leq 0~.\]
The argument is similar for $x = 0$.
\end{proof}

\begin{proof}[Proof of Lemma~\ref{l:7.5}]
Let $A \geq \log^{-1} (1+\epsilon)$. Then from Lemmata~\ref{l:equ} and \ref{l:convlog} the function $u_A = u - A \widetilde{\rho}$
satisfies
\[ u_A(x) \leq \sum_{y \sim x} u_A(y) \, \frac{g(y)}{2d(1+V(x))g(x)} - A \log\left(1 - \frac{1}{(1+V(0))g(0)}\right) \delta(x)~. \]
By a finite-volume approximation argument (which is applicable due to the deterministic bound (\ref{eq:determ})),
\[ \max u_A(x) = u_A(0) \leq - \frac{A}{1 - \frac{1}{(1+V(0))g(0)}} \log\left(1 - \frac{1}{(1+V(0))g(0)}\right) \leq A_\epsilon'~,\]
whence
\[ u(x) \leq A_\epsilon' + A \widetilde{\rho}(x) \leq C_\epsilon(1 + \rho(x))~.\]
\end{proof}

\begin{proof}[Proof of Proposition~\ref{p:bound.2}]
By Lemma~\ref{l:bound.infty}~,
\[\begin{split}
 L &= \sum_{y} \left[ \frac{G(0, y) G(y, x)}{G(0, x)} \right]^2 \\
   &\leq \max_y G(y, y) \,  \sum_{y} \frac{G(0, y) G(y, x)}{G(0, x)} = \max_y G(y, y) \, u(x)~.
\end{split}\]
The inequality $V \geq \epsilon$ implies $G(y, y) \leq A_\epsilon''$, and Lemma~\ref{l:7.5} implies
\[ u(x) \leq C_\epsilon (\rho(x) + 1)~. \]
\end{proof}

\vspace{3mm}\noindent
Next, we need
\begin{lemma}\label{l:conv}
For any $x \in \mathbb{Z}^d$, $\log g(x)$, $\log \frac{G(0,x)}{G(0,0)}$, and $\log \frac{G(0, x)}{G(x,x)}$
are convex functions of the potential. Consequently, 
\[ \rho(x) = - \frac{1}{2} \left[ \log \frac{G(0,x)}{G(0,0)} + \log \frac{G(0, x)}{G(x,x)} \right] \]
is a concave function of the potential.
\end{lemma}

\begin{proof}
The first statement follows from the random walk expansion:
\begin{equation}\label{eq:rw} g(x) = \sum \frac{1}{1 + V(x_0)} \frac{1}{2d} \frac{1}{1+V(x_1)} \frac{1}{2d} \cdots \frac{1}{2d} \frac{1}{1+V(x_k)}~,
\end{equation}
where the sum is over all paths $w: x_0 = 0, x_1, \cdots, x_{k-1}, x_k = x$. Indeed, for every $w$
\[ T_w = \log \frac{1}{1 + V(x_0)} \frac{1}{2d} \frac{1}{1+V(x_1)} \frac{1}{2d} \cdots \frac{1}{2d} \frac{1}{1+V(x_k)} \]
is a convex function of $V$, hence also $\log g(x) = \log \sum_w e^{T_w}$ is convex.

To prove the second statement, observe that
\[ G(0, x) = \frac{1}{2d} G(0, 0) \sum_{y \sim 0} G_0(y, x)~,\]
where $G_0$ is obtained by deleting the edges adjacent to $0$. Therefore
\[ \log \frac{G(0,x)}{G(0,0)} = - \log(2d) + \log \sum_{y \sim 0}  G_{0}(y, x)~; \]
for every $y$, $\log G_{0}(y, x)$ is a convex function of $V$, hence so is $\log \frac{G(0,x)}{G(0,0)}$.
\end{proof}

\begin{proof}[Proof of Theorem~\ref{thm:fluct}]
Denote $\rho_0(x) = \min(\rho(x), \mu(x))$. Then by Lemma~\ref{l:deriv} and Proposition~\ref{p:bound.2}
\[ \| \nabla_V \rho_0(x) \|_2^2 \leq A_\epsilon(\mu(x) + 1)~,\]
$A_\epsilon$ depends only on $\epsilon$. By Lemma~\ref{l:conv}, $\rho_0$ is concave, therefore by Lemma~\ref{l:radek}
\[ \mathbb{P} \left\{ \rho(x) \leq \mu(x) - t \right\}
    \leq \exp \left\{ - \frac{t^2}{C M^2 A_\epsilon (\mu(x)+1)} \right\}~. \]
Similarly, set $\rho_t(x) = \min(\rho(x), \mu(x) + t)$. Then
\[  \| \nabla_V \rho_t(x) \|_2^2 \leq A_\epsilon(\mu(x) + t + 1)~,\]
therefore by Lemma~\ref{l:radek}
\[\begin{split}
 \mathbb{P} \left\{ \rho(x) \geq \mu(x) + t \right\}
    &= \mathbb{P} \left\{ \rho_t(x) \geq \mu(x) + t \right\} \\
    &\underset{\triangle}{\leq} \exp \left\{ - \frac{t^2}{C M^2 A_\epsilon (\mu(x)+ t + 1)} \right\}~. 
\end{split}\]
\end{proof}

\section{Proof of Theorem~\ref{thm:bks}}

The proof follows the strategy of Benjamini, Kalai, and Schramm \cite{BKS}. Without loss of generality
we may assume that $\|x \| \geq 2$; set $m = \lfloor\| x \|^{1/4} \rfloor + 1$. 

To implement the Benjamini--Kalai--Schramm averageing argument, set
\[ F = - \frac{1}{\# B} \sum_{z \in B} \log G(z, x+z)~, \]
where 
\[ B = B(0, m) = \{ z \in \mathbb{Z}^d \, \mid \, \|z \|\leq m \} \]
is the ball of radius $m$ about the origin (cf.\ Alexander and Zygouras \cite{AZ}). According to Lemma~\ref{l:bound.infty},
\[ G(0, x) \geq \frac{G(z,x+z) G(0, z) G(x, x+z)}{G(z,z) G(x+z,x+z)}~, \]
therefore $\rho(x) \leq F + C_{a,b} m$; similarly, $\rho(x) \geq F - C_{a,b}m$. It is therefore sufficient to
show that 
\[ \mathrm{Var} \, F \leq C_{a,b} \, \frac{\|x\|}{\log \|x\|}~. \]

We use another inequality due to Talagrand \cite{T1} (see Ledoux \cite{L} for a semigroup derivation). Let $\mathcal{X}$ be
a (finite or countable) set. Let $\sigma_x^+: \{a,b\}^\mathcal{X} \to \{a,b\}^\mathcal{X}$ be the map setting the $x$-th
coordinate to $b$, and $\sigma_x^-: \{a,b\}^\mathcal{X} \to \{a,b\}^\mathcal{X}$ --the map setting the $x$-th
coordinate to $a$. Denote 
\[ \partial_x f = f \circ \sigma_x^+ - f \circ \sigma_x^-~. \]
\begin{lemma}[Talagrand]
For any function $f$ on $\{a, b\}^\mathcal{X}$,
\begin{equation}\label{eq:tal} \mathrm{Var} \, f \leq C_{a,b} 
  \sum_{x \in \mathcal{X}} \frac{\mathbb{E} |\partial_x f|^2}{1 + \log \frac{\mathbb{E} |\partial_x f|^2}{(\mathbb{E} |\partial_x f|)^2}}~. 
\end{equation}
\end{lemma}

Let us estimate the right-hand side for $f = F$, $\mathcal{X} = \mathbb{Z}^d$. Denote 
\[ \sigma_x^t = t \sigma_x^+ + (1-t) \sigma_x^-~; \]
then
\[ \partial_x F = \int_0^1 \frac{\partial F}{\partial V(x)} \circ \sigma_x^t \, dt~.\]
According to Lemma~\ref{l:deriv},
\[ \frac{\partial F}{\partial V(y)} = \frac{1}{\# B} \sum_{z \in B} \frac{G(z, y) G(y, x+z)}{G(z, x+z)}~. \]
Therefore
\[\begin{split} \mathbb{E} \frac{\partial F}{\partial V(y)} \circ \sigma_y^t 
  &= \mathbb{E} \frac{1}{\# B} \sum_{z \in B} \frac{G(z, y) G(y, x+z)}{G(z, x+z)} \circ \sigma_y^t \\
  &= \mathbb{E} \frac{1}{\# B} \sum_{z \in B} \frac{G(0, y-z)G(y-z,x)}{G(0, x)} \circ \sigma_{y-z}^t  \\
  &= \mathbb{E} \frac{1}{\# B} \sum_{v \in y + B} \frac{G(0, v)G(v,x)}{G(0, x)} \circ \sigma_v^t~.
\end{split}\]

\begin{lemma}\label{l:bks} For any $Q \subset \mathbb{Z}^d$ and any $x',x \in \mathbb{Z}^d$,
\begin{equation}\label{eq:lbks}
\sum_{v \in Q} \frac{G(x',v)G(v,x)}{G(x',x)} \leq C_a (\mathrm{diam}_\rho \, Q + 1) \leq C_{a,b} (\mathrm{diam} \, Q + 1)~. 
\end{equation}
\end{lemma}

Let us first conclude the proof of Theorem~\ref{thm:bks} and then prove the lemma. 
Set $\delta = m^{-\frac{1}{2}}$, and let 
\[ A = \left\{ y \in \mathbb{Z}^d \, \Big| \, \mathbb{E} \left( \partial_y F \right)^2 
  \leq \delta \, \mathbb{E} \partial_y F  \right\}~. \]
Then the contribution of coordinates in $A$ to the right-hand side of (\ref{eq:tal}) is at most $C \delta \|x\|$
by Lemma~\ref{l:7.5}. For $y$ in the complement of $A$, Lemma~\ref{l:bks} yields
\[ \mathbb{E} \partial_y F \leq \frac{Cm}{\# B}~, \]
hence 
\[ \mathbb{E} \left(\partial_y F \right)^2 \geq \delta \, \mathbb{E} \partial_y F
  \geq \frac{\delta \# B}{C m} \left( \mathbb{E} \partial_y F \right)^2~,\]
and
\[ \log \frac{\mathbb{E} \left(\partial_y F\right)^2}{ \left( \mathbb{E} \partial_y F \right)^2}
  \geq \log \frac{\delta}{Cm} \geq  \log (\|x\|/C') \]
by the inequality $\# B \geq Cm^2$ (which holds with $d$-independent $C$). The contribution of the complement of $A$ to (\ref{eq:tal})
is therefore at most $C' \frac{\|x\|}{\log \|x\|}$. Thus finally
\[ \mathrm{Var} F \leq \frac{C''\|x\|}{\log \|x\|}~. \]
\qed

\begin{proof}[Proof of Lemma~\ref{l:bks}]
For $Q \subset \mathbb{Z}^d$ and $x',x \in \mathbb{Z}^d$, set
\[ u_Q(x',x) = \frac{(G\mathbbm{1}_QG)(x',x)}{G(x',x)} = \frac{\sum_{q \in Q} G(x', q) G(q, x)}{G(x',x)}~. \]
Similarly to Lemma~\ref{l:equ},
\[ u_Q(x',x) = \sum_{y\sim x} u_Q(x', y) \frac{G(x',y)}{2d(1+V(x))G(x',x)} + \frac{\mathbbm{1}_Q(x)}{1+V(x)}~.\]
By a finite-volume approximation argument, it is sufficient to prove the estimate (\ref{eq:lbks}) in a finite box. Then
$\max_{x} u_Q(x', x)$ is attained for some $x_{\max} \in Q$. By symmetry, $\max_{x',x} u_Q(x',x)$ is attained when
both $x'$ and $x$ are in $Q$. On the other hand, for $x',x \in Q$
\[ u_Q(x',x) \leq u_{\mathbb{Z}^d}(x',x) \leq C(1 + \log\frac{1}{G(x',x)}) \leq C'(1 + \mathrm{diam_\rho} \, Q) \]
by Lemma~\ref{l:7.5}.
\end{proof}

\begin{rmk}
To extend Theorem~\ref{thm:bks} to the generality of the work of Bena\"im and Rossignol \cite{BR}, one may
use the modified Poincar\'e inequality of \cite{BR} instead of Talagrand's inequality (\ref{eq:tal}); this argument also yields a tail bound as in \cite{BR}. 
\end{rmk}

One may also hope that the even more general methods of Damron, Hanson, and Sosoe \cite{DHS} could 
be adapted to the setting of the current paper.

\section{A remark}\label{s:rem}

Let $H = -\frac{1}{2d} \Delta + V$ be a random Schr\"odinger operator
on $\mathbb{Z}^d$. For $z \in \mathbb{C} \setminus \mathbb{R}$, set
$G_z = (H - z)^{-1}$. The analysis of  Somoza, Ortu\~no, and Prior \cite{SOP}
(see further Le Doussal \cite{LD}) suggests that, in dimension $d= 2$,
\[ \operatorname{Var} \log |G_z(x, y)| \asymp (\| x - y \| + 1)^{2/3}~,\]
and that a similar estimate is valid for the boundary values
$G_{\lambda + i0}$ (which exist for almost every $\lambda \in \mathbb{R}$) even when $\lambda$
is in the spectrum of $H$.

Having this circle of questions in mind, it would be interesting to study
the fluctuations of $\log |G_z(x, y)|$ for  $z \in \mathbb{C} \setminus \mathbb{R}$. 
In dimension $d = 1$, the results of Furstenberg and Kesten \cite{FK}  imply that
\[ \operatorname{Var} \log |G_z(x, y)| \asymp |x-y| + 1~.\]
We are not aware of any rigorous bounds in dimension $d \geq 2$. In
particular, we do not know a proof of the estimate
\begin{equation}\label{eq:qn}
\operatorname{Var} \log |G_z(x, y)| = o(\|x- y\|^2)~,
\quad \|x-y\| \to \infty~,
\end{equation}
even when $z$ is such that the random walk representation
(\ref{eq:rw}) is convergent.

\paragraph{Acknowledgment.} I am grateful to Thomas Spencer for helpful conversations, and to Itai Benjamini,
Michael Damron, Alexander Elgart, and Gil Kalai for their comments on a preliminary version of this note.

\end{document}